\PassOptionsToPackage{nameinlink}{cleveref}
\documentclass[cleveref, thm-restate]{lipics-v2021}

\usepackage{graphicx}
\usepackage{microtype}

\usepackage{cite}
\usepackage{amsmath, amsfonts, mathrsfs}

\newcommand{\F}{\mathcal{F}}

\newcommand{\EA}{\text{EA}}
\newcommand{\LD}{\text{LD}}
\newcommand{\LA}{\text{LA}}

\renewcommand{\a}{\text{arrival}}
\renewcommand{\d}{\text{departure}}
\renewcommand{\succ}{\text{succ}}
\newcommand{\pred}{\text{pred}}

\usepackage{makecell}

\title{Temporal queries for dynamic temporal forests}



\author{Davide Bilò}{Department of Information Engineering, Computer Science, and Mathematics, University~of~L'Aquila, Italy}{davide.bilo@univaq.it}{https://orcid.org/0000-0003-3169-4300}{}

\author{Luciano Gualà}{Department of Enterprise Engineering, University of Rome ``Tor Vergata'', Italy}{guala@mat.uniroma2.it}{https://orcid.org/0000-0001-6976-5579}{}

\author{Stefano Leucci}{Department of Information Engineering, Computer Science, and Mathematics, University~of~L'Aquila, Italy}{stefano.leucci@univaq.it}{https://orcid.org/0000-0002-8848-7006}{}

\author{Guido Proietti}{Department of Information Engineering, Computer Science, and Mathematics, University~of~L'Aquila, Italy}{guido.proietti@univaq.it}{https://orcid.org/0000-0003-1009-5552}{}

\author{Alessandro Straziota}{Department of Enterprise Engineering, University of Rome ``Tor Vergata'', Italy}{alessandro.straziota@uniroma2.it}{https://orcid.org/0009-0008-4543-786X}{}

\authorrunning{D. Bilò, L. Gualà, S. Leucci, G. Proietti, and A. Straziota} 

\Copyright{Davide Bilò, Luciano Gualà, and Stefano Leucci, Guido Proietti, and Alessandro Straziota} 

\EventEditors{Juli\'{a}n Mestre and Anthony Wirth}
\EventNoEds{2}
\EventLongTitle{35th International Symposium on Algorithms and Computation (ISAAC 2024)}
\EventShortTitle{ISAAC 2024}
\EventAcronym{ISAAC}
\EventYear{2024}
\EventDate{December 8--11, 2024}
\EventLocation{Sydney, Australia}
\EventLogo{}
\SeriesVolume{322}
\ArticleNo{47}

\hideLIPIcs
\nolinenumbers
\begin{document}

\maketitle

\begin{abstract}
In a temporal forest each edge has an associated set of time labels that specify the time instants in which the edges are available. A temporal path from vertex $u$ to vertex $v$ in the forest is a selection of a label for each edge in the unique path from $u$ to $v$, assuming it exists, such that the labels selected for any two consecutive edges are non-decreasing.

We design linear-size data structures that maintain a temporal forest of rooted trees under addition and deletion of both edge labels and singleton vertices, insertion of root-to-node edges, and removal of edges with no labels.
Such data structures can answer temporal reachability, earliest arrival, and latest departure queries. All queries and updates are handled in polylogarithmic worst-case time.
Our results can be adapted to deal with latencies. More precisely, all the worst-case time bounds are asymptotically unaffected when latencies are uniform. For arbitrary latencies, the update time becomes amortized in the incremental case where only label additions and edge/singleton insertions are allowed as well as in the decremental case in which only label deletions and edge/singleton removals are allowed. 

To the best of our knowledge, the only previously known data structure supporting temporal reachability queries is due to Brito, Albertini, Casteigts, and Travençolo [Social Network Analysis and Mining, 2021], which can handle general temporal graphs, answers queries in logarithmic time in the worst case, but requires an amortized update time that is quadratic in the number of vertices, up to polylogarithmic factors.
\end{abstract}

\section{Introduction}

\emph{Temporal graphs}, a.k.a.\ time-varying graphs, are graphs in which each edge is only available in specific time instants. 
These graphs are crucial for modeling and analyzing a variety of complex systems such as social networks, communication networks, transportation systems, and biological networks, where relationships and interactions evolve. 

A simple and widely-adopted model of temporal graphs is the one of Kempe, Kleinberg, and Kumar~\cite{KempeKK02} in which a temporal graph is modeled as a graph $G$ where each edge has an assigned integral temporal label representing a time instant in which the edge is available. Temporal paths are time-respecting walks on $G$, i.e., walks in which any two consecutive traversed edges have non-decreasing labels.
This model naturally generalizes to the case in which an edge may be traversed in multiple time instants, i.e., when each edge of $G$ is associated with a set of integer labels, and
many other generalizations have also been proposed (see, e.g., \cite{Holme14}).

Due to their generality, many algorithmic problems regarding temporal graphs have been considered in the literature, including the computation of good temporal paths \cite{WuCHKLX14,CasteigtsHMZ21,AkridaMNRSZ20} w.r.t.\ several quality measures, sparsification \cite{KempeKK02,CasteigtsPS21,BiloDG0R22,BiloDGLR24,CasteigtsRRZ24}, exploration \cite{ErlebachKLSS19,ErlebachS23}, a temporal version of the classical vertex cover problem \cite{AkridaMSZ20,HammKMS22}, and network creation games \cite{BiloC0GKLS23}, to mention~a~few. 
Perhaps surprisingly, the problem of designing dynamic data structures maintaining information on temporal connectivity has only been addressed quite recently, despite the vast amount of literature for the non-temporal case (see, e.g., \cite{HuangHKPT23,HanauerHS22} and the references therein). In this regard, Brito, Albertini, Casteigts, and Travençolo~\cite{Brito2021ADD} provide a data structure that can answer \emph{temporal reachability} queries in the \emph{incremental} setting, i.e., under the addition of new edges or of new labels to already existing edges. This data structure has size $O(n^2 \tau)$, and an amortized update time of $O(n^2 \log \tau)$, where $n$ is the number of vertices of the temporal graph and $\tau$ is its \emph{lifetime}, i.e., the number of distinct labels. 
Given two vertices $u,v$ and two time instants $t,t'$, it can report whether there exists a temporal path from $u$ to $v$ that departs from $u$ no earlier than $t$ and arrives in $v$ no later than $t'$. The worst-case time required by such a query is $O(\log \tau)$ and the corresponding path can be retrieved in $O(\log \tau)$ worst-case time per edge.

In \cite{BarjonCCJN14} the authors consider the problem of maintaining a directed graph in which edge $(u,v)$ exists if there is a temporal path from $u$ to $v$ in $G$, and show how to efficiently update the \emph{time-transitive closure} in the restricted \emph{chronological incremental} case, where the inserted labels must be monotonically non decreasing.

\subparagraph*{Our results.}
We focus on the case in which the temporal graph of interest is acyclic, i.e., it is a temporal forest $F$, and we design data structures that support insertions and deletions of edge labels and can answer temporal reachability queries along with the more general \emph{earliest arrival} (EA) and \emph{latest departure} (LD) queries. Given two distinct vertices $u$, $v$ and time instant $t$, an EA (resp.\ LD) query reports the smallest arrival time (reps.\ largest departure time) among those of all temporal paths from $u$ to $v$ that depart no earlier than $t$ (resp.\ arrive no later than $t$).

Our main data structure is dynamic also w.r.t.\ a second aspect. In addition of being able to add and remove labels from the edges of $F$, it also supports the addition and deletion of singleton vertices along with \emph{link} and \emph{cut} operations in a temporal forest of rooted trees. The link operation adds a new edge $(u,v)$ with label $\ell$ in $F$ from the root $u$ of any tree to any other vertex $v$ in a different tree, thus merging the two trees into one.
The cut operation deletes the last label of an edge $e$ and removes $e$ from $F$, thus splitting the tree $T$ that contained $e$ into two trees $T_1, T_2$, where the root of $T_1$ coincides with that of $T$ and $T_2$ is rooted in the unique endvertex of $e$ in $T_2$. 

Our data structure has linear size, and can be updated in $O(\log M)$ worst-case time per operation, where $M$ is the total number of (not necessarily distinct) labels in the forest at the time of the operation. EA and LD queries require $O(\log L \cdot \log M)$ worst-case time, where $L$ is the maximum number of labels on a single edge at the time of the operation, and a temporal reachability query can be answered in $O(\log M)$ worst-case time (see \Cref{thm:fully_dynamic_forest}).  

In the special case of temporal paths with fixed topology (i.e., when only insertions and deletions of edge labels are allowed), we can improve the worst-case time complexity of EA and LD queries to $O(\log M)$, thus shaving a $O(\log L)$ factor (see \Cref{thm:static_path}).

Our data structure can also be adapted to the more general case of temporal forests with \emph{latencies}, in which the labels on a generic edge are pairs $(\ell, d)$, with $d\geq 0$, encoding the fact that the edge can be traversed at time $\ell$ from any endvertex to reach the other endvertex at time $\ell + d$.
The data structure retains the same worst-case asymptotic guarantees when latencies are \emph{uniform} (see \Cref{cor:uniform_latencies}).
For arbitrary latencies, we consider both the \emph{incremental} and the \emph{decremental} scenarios.
In the incremental scenario we support the addition of new labels to existing edges, the addition of singleton vertices, and the above link operation, while the decremental scenario only involves label/singleton deletions and cut operations.
In both cases, the above $O(\log M)$ upper bound on the update time becomes amortized (see \Cref{thm:incremental_decremantal_latencies}).

\section{Preliminaries}
A \emph{temporal forest} is an undirected forest $F = (V(F),E(F))$ paired with a function $\lambda: E(F) \to \mathscr{P}(\mathbb{Z})$ that associates a set $\lambda(e) \subseteq \mathbb{Z}$ of labels to each edge $e \in E(F)$.
A \emph{temporal path} $\pi$ from vertex $u\in V(F)$ to vertex $v \in V(F) \setminus \{u\}$ in $F$ is a sequence of pairs $\langle (e_1, \ell_1),  (e_2, \ell_2), \dots, (e_k, \ell_k) \rangle$ such that $\langle e_1, e_2, \dots, e_k \rangle$ is a path from $u$ to $v$  in $F$, $\ell_i \in \lambda(e_i)$ for all $i=1,\dots,k$, and $\ell_i \le \ell_{i+1}$ for all $i=1, \dots, k-1$.
The departure time $\d (\pi)$ and the arrival time $\a (\pi)$ of $\pi$ are $\ell_1$ and $\ell_k$, respectively.

Given two distinct vertices $u, v \in V(F)$, and $t_a,t_d \in \mathbb{Z} \cup \{-\infty, +\infty\}$,  $\Pi_{F}(u, v, t_d, t_a)$ is the set of all temporal paths from $u$ to $v$ in $F$ having a departure time of at least $t_d$ and an arrival time of at most $t_a$. 
For a time $t \in \mathbb{Z} \cup \{-\infty\}$, the \emph{earliest arrival time} of a temporal path from $u$ to $v$ with departure time at least $t$ is denoted by $\EA(u, v, t)$ and is defined as  $\min_{\pi \in \Pi(u, v, t, +\infty)} \a(\pi)$. If no temporal path from $u$ to $v$ with departure time at least $t$ exists in $F$, i.e., if $\Pi(u, v, t, +\infty) = \emptyset$, we define $\EA(u,v,t) = +\infty$. A path $\pi \in \Pi(u,v,t,+\infty)$ such that $\a(\pi)=\EA(u,v,t)$ is called an \emph{earliest arrival path}.

Similarly, we denote the \emph{latest departure time} of a temporal path from $u$ to $v$ having an arrival time of at most $t \in \mathbb{Z} \cup \{+\infty\}$ as $\LD(u, v, t) =  \max_{\pi \in \Pi(u, v, -\infty, t)} \d(\pi)$.  If $\Pi(u, v, -\infty, t) = \emptyset$, then we let $\LD(u,v,t) = -\infty$. A path $\pi \in \Pi(u,v,t,+\infty)$ such that $\a(\pi)=\LD(u,v,t)$ is called a \emph{latest departure path}. As an edge case, we define $\EA(u, u, t) = \LD(u, u, t) = t$ for all $u \in V(F)$.

We design a data structure for maintaining information on a dynamic temporal forest of rooted trees. The data structure must efficiently answer to the following queries: 

\begin{description}
    \item[Earliest arrival time (EA query):] Given two vertices $u,v \in V(F)$ with $u \neq v$, and a time $t \in \mathbb{Z} \cup \{-\infty\}$, report $\EA(u,v,t)$;

    \item[Latest departure time (LD query):] Given two vertices $u,v \in V(F)$ with $u \neq v$, and a time $t \in \mathbb{Z} \cup \{+\infty\}$, report $\LD(u,v,t)$;

    \item[Temporal reachability:]
    Given two vertices $u,v \in V(F)$ with $u \neq v$, and two times $t_a, t_d \in \mathbb{Z} \in \{ -\infty, +\infty \}$ with $t_a \leq t_d$, report whether there exists a temporal path $\pi$ from $u$ to $v$ in $F$ with $\d(\pi) \geq t_a$ and $\a(\pi) \leq t_d$.
\end{description}

Clearly, a data structure that supports EA queries or LD queries also supports temporal reachability queries with the same worst-case query time. 
The update operations we consider are the following:

\begin{description}
    \item[Label addition:] Given an edge $e \in E(F)$ and a value $\ell \in \mathbb{Z} \setminus \lambda(e)$, add $\ell$ to $\lambda(e)$;

    \item[Label deletion:] Given an edge $e \in E(F)$ with $|\lambda(e)|\geq 2$ and some $\ell \in \lambda(e)$, delete $\ell$ from $\lambda(e)$;

    \item[Link:] Given two vertices $u,v$ where $u$ is the root of some tree $T$ in $F$ and $v$ is some vertex of a tree $T'$ of $F$ with $T' \neq T$, and a label $\ell \in \mathbb{Z}$, add the edge $(u,v)$ to $F$ and set $\lambda((u,v))=\{\ell\}$. The new tree resulting from merging $T$ with $T'$ retains the root of $T'$;

    \item[Cut:] Given an edge $e \in E(F)$ with $|\lambda(e)|=1$, remove the edge $e$ from the tree $T$ of $F$ containing $e$, thus splitting $T$ into two new trees $T_1, T_2$ where $T_1$ retains the root of $T$ and $T_2$ is rooted in the unique endvertex of $e$ in $T_2$; 

    \item[Singleton addition/deletion:] Add a new (resp.\ remove an existing) singleton vertex to (resp.\ from) $F$. 
\end{description}

If $u,v,w$ are vertices of $F$ such that $w$ lies on the unique path between $u$ and $v$ in $F$ we have  that
$\EA(u, v, t) = \EA(w, v, \EA(u, w, t))$ and $\LD(u, v, t) = \LD(u, w, \LD(w, v, t))$.

Let $-F$ be the temporal forest such that $V(F)=V(-F)$, $E(F)=E(-F)$, and for every $e \in E(F)$, the set of labels of $e$ in $-F$ is $\{-\ell \mid \ell \in \lambda(e)\}$. One can observe that:
\begin{lemma}\label{lemma:LD_eq_minusLA}
    The value of $\LD(u,v,t)$ in $F$ coincides with $- \EA(v,u,-t)$ in $-F$.
\end{lemma}

Finally, given a subset $X$ of a universe possessing a strict total order, and an element $x$ of the universe, we define the \emph{successor} $\succ(x,X)$  (resp. \emph{predecessor} $\pred(x,X)$) of $x$ w.r.t.\ $X$ as the minimum element $y \in X$ s.t.\ $y \geq x$ (resp. the maximum element $y \in X$ s.t.\ $y \leq x$). If such an element $y$ does not exist, $\succ(x, X)=+\infty$. 
Similarly, the \emph{strict successor} (resp. \emph{strict predecessor}) is defined as $\succ(x,X\setminus \{x\})$ (resp. $\pred(x,X\setminus \{x\})$).

\section{Warm up: a data structure for temporal paths}\label{sec:paths}

As a warm up, in this section we consider the simpler case in which the forest $F$ is actually a fixed path, and the only supported operations are label additions and deletions to/from the edges of $F$.
We use this section to introduce some of the ideas of our construction, which we generalize to the case of dynamic temporal forests in \Cref{sec:temporal_forests} that also contains the correctness proofs. More precisely, we establish the following result:

\begin{theorem}\label{thm:static_path}
    Given a temporal path with $n$ vertices, it is possible to build in $O(n+M \log L)$ time, a data structure of linear size supporting EA and LD queries under label insertions and deletions in $O(\log M)$ worst-case time per operation, where $L$ is the maximum number of labels on the same temporal edge, and $M$ is the total number of (not necessarily distinct) labels in the path at the time of the operation.
\end{theorem}

Fix an arbitrary endvertex $v_0$ of the path $F$, and let $v_i$ be the unique vertex at hop-distance $i$ from $v_0$ in $F$.
We think of $F$ as a tree rooted in $v_{n-1}$. With this interpretation, the main technical difficulty lies in designing a data structure capable of answering \emph{upward} EA queries, i.e., EA queries from some vertex $v_i$ to some vertex $v_j$ with $j > i$.
The cases of \emph{downward} EA queries can be managed by a similar data structure rooted in $v_0$. LD queries and reachability queries can be recast as EA queries as already discussed in the preliminaries and in \Cref{lemma:LD_eq_minusLA}.

The naive solution to represent all possible upward earliest arrival paths in $F$ is that of building a forest of rooted trees similar to the one shown on the top-left of \Cref{fig:enter-label}. This forest represents each label $\ell$ of each edge $e = (v_i, v_{i+1})$ in $F$ as a node whose parent corresponds to the first label of $(v_{i+1}, v_{i+2})$ that is larger than or equal to $\ell$ (if any).
Intuitively, moving upwards in $F$ along an earliest arrival path corresponds to moving upwards in a tree of such a forest. Unfortunately, explicitly maintaining this forest would be costly since even a single label insertion/deletion could cause a large number of edges to be rewired. 
We circumvent this problem by maintaining a different edge-weighted forest $\F(\lambda)$ that still represents all possible upward earliest arrival paths while guaranteeing that each node has constant degree. 

Let $e_i$ be the edge $(v_i, v_{i+1})$.  
Our forest $\F(\lambda)$ contains a \emph{node} $(\ell, i)$ for each edge $e_i$ in $F$, and for each $\ell \in \lambda(e_i)$.
To aid readability we use the term \emph{vertex} to refer to vertices in $F$, and the term \emph{node} for those in $\F(\lambda)$. To describe the edges of $\F(\lambda)$, we first define a partial function $\sigma_i(\ell, \lambda)$ that maps the labels $\ell \in \lambda(e_i)$ to a ``successor'' node, which is either the one corresponding to the smallest label strictly larger than $\ell$ on the same edge $e_i$, or the one corresponding to the successor of $\ell$ among the labels of edge $e_{i+1}$.
More precisely, given  $i \in \{0, \dots, n-2\}$ and $\ell \in \lambda(e_i)$, we consider $\ell^+ = \succ(\ell+1, \lambda(e_i) )$. For $i=n-2$, $\sigma_i(\ell, \lambda)$ is defined as $(e_i, \ell^+)$ if $\ell^+ < +\infty$, and is undefined if $\ell^+ = +\infty$.
For $i \in \{0, \dots, n-3\}$, we compare $\ell^+$ with $\ell' = \succ(\ell, \lambda(e_{i+1}))$. 
If $\ell^+ = \ell' = +\infty$, then $\sigma_i(\ell, \lambda)$ is undefined. 
Otherwise
$\sigma_i(\ell, \lambda) = 
\begin{cases}
(\ell^+, i) & \text{if } \ell^+ \le \ell'; \\
(\ell', i+1) & \text{if } \ell' < \ell^+.
\end{cases}$

All the nodes $(\ell, i)$ for which $\sigma_i(\ell, \lambda)$ is undefined are the roots of their respective trees in $\F(\lambda)$. 
Whenever $\lambda$ is clear from context we write $\F$ in place of $\F(\lambda)$ and $\sigma_i(\ell)$ in place of $\sigma_i(\ell, \lambda)$. See \Cref{fig:enter-label} (top-right) for an example.

The weights of all edges of the form $( (\cdot, i), (\cdot, i+1) )$ are set to $1$, while those of the remaining edges, of the form $((\cdot, i), (\cdot, i))$, are set to $0$.
Our definition of $\sigma_i$ ensures that different children of the same node $v$ in $\F$ must be linked to $v$ using edges of different weights. As a consequence, each vertex has at most $2$ children in $\F$. 

\begin{figure}[t]
    \centering
    
    \includegraphics[width=0.45\textwidth]{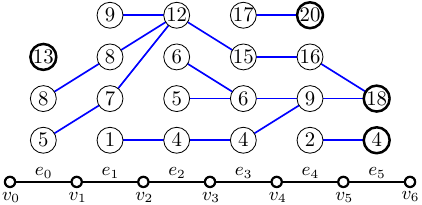}
    \hspace{0.05\textwidth}
    \includegraphics[width=0.45\textwidth]{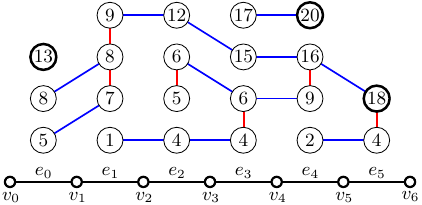} \\[.2cm]
    \includegraphics[width=0.45\textwidth]{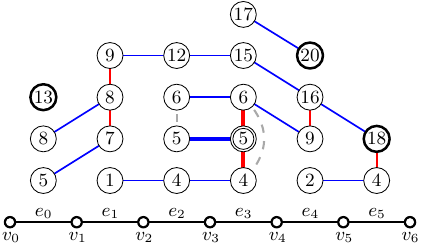}
    \hspace{0.05\textwidth}
    \includegraphics[width=0.45\textwidth]{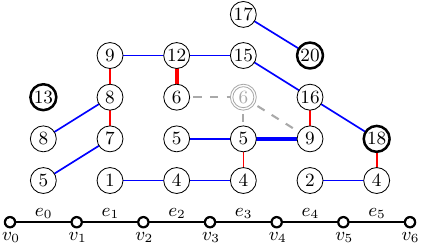}

    \caption{Top left: a forest representing all possible upward earliest arrival paths. 
    Top right: the corresponding weighted forest $\F$. Red edges have weight $0$, while blue ones have weight $1$. Here, a label $\ell \in \lambda(e_i)$ is shown above edge $e_i$ and models the node $(\ell, i)$. 
    Bottom left and bottom right: the forests $\F$ obtained after adding label $5$ to $\lambda(e_3)$ and then deleting label $6$ from $\lambda{e_3}$.  New edges are shown in bold while removed ones are dashed.}
    \label{fig:enter-label}
\end{figure}

\noindent Our data structure stores:
\begin{itemize}
    \item A \emph{top tree} \cite{AlstrupHLT05} representing $\F$. A top-tree is a data structure capable of maintaining an edge-weighted forest under insertion of new nodes,  deletion of singleton nodes, and under link and cut operations, i.e., addition and deletion of edges.
    Moreover, it supports weighted level ancestor queries: given a vertex $v$ and $w \in \mathbb{N}$, it reports the $w$-th weighted level ancestor $\LA(v, w)$ of $v$ in $\F$, i.e., the deepest ancestor of $v$ at distance at least $w$ from $v$ (if any).
    Each of the above operations requires $O(\log \eta)$ worst-case time, where $\eta$ is the number of nodes in the forest. A top tree with $\eta$ nodes can be built in $O(\eta)$ time. Whenever a node/edge is added removed from $\F$, we perform the corresponding update operation on the backing top tree.
    
    \item A dictionary $D_i$ for each edge $e_i$ that supports insertions, deletions, and \linebreak predecessor/successor queries in $O(\log \eta)$ worst-case time per operation, where $\eta$ is the number of keys in $D_i$. The dictionary can be build in $O(\eta \log \eta)$ time. The keys in $D_i$ are the labels in $\lambda(e_i)$ and each label $\ell$ is associated with a pointer to vertex $(\ell, i)$ in $\F$. Such a dictionary can be implemented using any dynamic balanced binary search tree.
\end{itemize}

\subparagraph*{Answering EA queries.}
To report $\EA(v_i, v_j, t)$ with $j>i$, we first find $\ell = \succ(t, \lambda(e_i))$ using $D_i$, then we query $\F$ for the $(j-i-1)$-th weighted level ancestor $(e_{j-1}, \ell^*)$ of $(e_i, \ell)$.
If such ancestor exist, we answer with $\ell^*$. Otherwise we answer with $+\infty$. Finding $\ell$ requires time $O(\log L)$ in the worst case, while querying $\F$ can be done in worst-case $O(\log M)$ time. Hence the overall worst-case time required to answer the query is $O(\log M)$.

\subparagraph*{An auxiliary procedure.}
A label insertion causes the addition of a new node into $\F$, while a label deletion causes the deletion of the corresponding node from $\F$.
This may result in some nodes in $\F$ that have an incorrect parent (or lack of thereof) from the one that would be expected according to the construction discussed above.

To address this problem, we define an auxiliary procedure FixParent($\ell,  i$) that will be helpful in restoring the correct state of $\F$. Such procedure takes the index $i$ of some edge $e_i$ in $F$ and a label $\ell \in \lambda(e_i)$, and ensures that the edge from node $(\ell, i)$ to its parent in $\F$ (if any) is properly set (or unset) by only considering $D_{i}$ and $D_{i+1}$.

To implement FixParent($\ell,  i$) we first \emph{cut} the current edge from $(\ell, i)$ to its parent in $\F$ (if any).
Then, we compute $\sigma_i(\ell)$ in $O(\log L)$ time by searching for the needed successors in $D_i$ and $D_{i+1}$, according to the definition of $\sigma_i$. Finally, we \emph{link} vertex $(\ell, i)$ with $\sigma_i(\ell)$, if any, in $\F$.
Since link and cut operations also require time $O(\log M)$, the overall worst-case time spent by FixParent is $O(\log M)$.

\subparagraph*{Label addition.}
To add label $\ell$ on edge $e_i$, we first insert node $(\ell, e_i)$ into $\F$, and $\ell$ into $D_i$. 
Then, if $i \ge 1$ and $\ell' = \pred(\ell, \lambda(e_{i-1}))$ exists, we perform FixParent($\ell, i-1$). Next, we perform FixParent($\ell, i$). Finally, if $\ell^- = \pred(\ell-1, \lambda(e_i))$ exists, we perform FixParent($\ell^-, i$). See \Cref{fig:enter-label} for an example.

The overall worst-case time required is $O(\log M)$ since we only perform a constant number of predecessor/successor lookups, node insertions into $\F$, and calls to FixParent.

\subparagraph*{Label deletion.}
To remove label $\ell$ from edge $e_i$, we first delete $\ell$ from $D_i$ and remove node $(e_i, \ell)$, along with all its incident edges, from $\F$.
Then, if $i \ge 1$ and $\ell' = \pred(\ell, \lambda(e_{i-1}))$ exists, we perform FixParent($\ell', i-1$).
Finally, if $\ell^- = \pred(\ell-1, \lambda(e_i))$ exists, we perform FixParent($\ell^-, i$). The overall worst-case time required is $O(\log M)$ since we only perform a constant number of predecessor/successor lookups, edge/node deletions from $\F$, and calls to FixParent.

\section{Our data structure for temporal forests}
\label{sec:temporal_forests}

In this section we discuss how our data structure for temporal paths with static topology can be generalized to temporal forests of rooted
trees while also supporting link and cut operations as well as insertions and deletions of singleton vertices.
More precisely, we prove the following result.

\begin{theorem}\label{thm:fully_dynamic_forest}
    Given a temporal forest of rooted trees, it is possible to build, in $O(n + M\log L)$ time, a data structure of linear size, where $n$ is the number of vertices, $L$ is the maximum number of labels on the same temporal edge, and $M$ is the total number of (not necessarily distinct) labels in the forest at the time of the operation. The data structure supports label additions, label deletions, link, and cut operations, and addition/deletion of singleton vertices in $O(\log M)$ worst-case time per operation, EA and LD queries in $O(\log L \cdot \log M)$ worst-case time per operation, and reachability queries in $O(\log M)$ worst-case time per operation. 
\end{theorem}

We observe that it is easy to extend \Cref{thm:static_path} to \emph{fixed} temporal forests where one only needs to support label additions and deletions and aims for a query time of $O(\log n \cdot \log M)$, where $n$ is the number of vertices in $F$.
Indeed, it suffices to construct the data structure for temporal paths on each path of a heavy-light decomposition \cite{SLEATOR1983362} of (each tree of) $F$. 
Details are given in \Cref{apx:heavy_light_decomposition}.

This section is organized as follows: we first discuss a data structure that only supports updates are label additions and deletions, and then we show how to handle link and cut operations as well as insertions and deletions of singleton vertices.

\subsection{A data structure supporting only label additions and deletions}\label{subsec:fixed_forest}

While the natural generalization of our construction for temporal paths of \Cref{sec:paths} to the case of trees would already provide a data structure supporting fast upward EA queries, such a solution runs into a similar problem as the one discussed for the naive approach. 
Indeed, in a tree of degree $\Delta$ a node of $\F$ may have $\Omega(\Delta)$ children, and its removal may cause the rewiring of $\Omega(\Delta)$ edges.

Our construction avoids this problem by further grouping the (nodes in $\F$ corresponding to the) edges from sibling vertices to their common parent $v$ in $F$ into a \emph{block} $\mathcal{B}_v$ (see \Cref{fig:tree}).

For a non-root vertex $v \in V(F)$, we denote by $p(v)$ the parent of $v$ in $F$, and by $e_v$ the edge $(v, p(v))$.
The forest $\F(\lambda)$ contains a \emph{node} $(\ell, v)$ for each non-root vertex $v \in V(F)$ and for each $\ell \in \lambda(e_v)$.
Moreover, for every non-leaf vertex $v \in V(F)$, we define the \emph{block of $v$} as the set $\mathcal{B}_v(\lambda)$ 
containing all nodes $(\ell, u)$ where $u$ is a child of $v$ in $F$.

We fix an arbitrary strict total ordering of the vertices of $V(F)$ and we think of the nodes of $\F$ as being ordered w.r.t.\ the order relation that compares nodes lexicographically, i.e., $(\ell, u)$ precedes $(\ell', v)$ if either $\ell < \ell'$, or if $\ell = \ell'$ and $u$ precedes $v$ in the chosen ordering of the vertices.

\begin{figure}
    \centering
    \includegraphics[width=\textwidth]{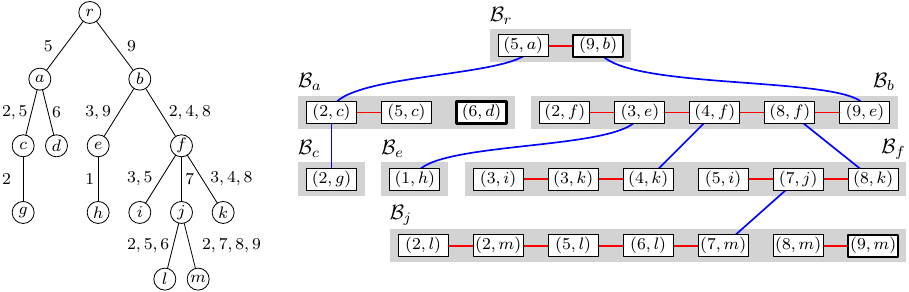}
    \caption{Left: a sample rooted temporal tree. Right: a representation of the forest $\F$ maintained by the data structure for temporal forests of \Cref{sec:temporal_forests}.}
    \label{fig:tree}
\end{figure}

Similarly to the case of temporal paths, the edges of $\F(\lambda)$ are defined by making use of a partial function $\sigma_{u}(\ell, \lambda)$ that associates the nodes of $\F$ to their parents.
More precisely, given a non-root vertex $u \in V(F)$ and $\ell \in \lambda(e_u)$, we let $v = p(u)$ so that $(\ell, u) \in \mathcal{B}_v$.
If $v$ is a root in $F$, $\sigma_{u}(\ell, \lambda)$ is defined as the strict successor of $(\ell, u)$ in $\mathcal{B}_v$, if it exists, otherwise $\sigma_{u}(\ell, \lambda)$ is undefined.

To define $\sigma_{u}(\ell, \lambda)$ when $v$ is not a root of $F$, let $\ell' = \succ(\ell, \lambda(e_v))$. We distinguish two cases depending on whether the strict successor $(\ell^+, u^+)$ of $(\ell, u)$ in $\mathcal{B}_v$ exists.
If $(\ell^+, u^+)$ exists we define
$
\sigma_u(\ell, \lambda) = 
\begin{cases}
        (\ell^+, u^+) & \text{if } \ell^+ \le \ell'; \\
        (\ell', v)  & \text{if } \ell^+ > \ell'.
\end{cases}
$

\noindent If $(\ell^+, u^+)$ does not exist, then $\sigma_u(\ell, \lambda) = (\ell', v)$ when $\ell' < +\infty$ and $\sigma_u(\ell, \lambda)$ is undefined when $\ell' = + \infty$. The nodes $(\ell, u)$ such that $\sigma_u(\ell , \lambda)$ is undefined are the roots of the corresponding trees in $\F(\lambda)$.
Whenever $\lambda$ is clear from context, we write $\F$ in place of $\F(\lambda)$, $\mathcal{B}_v$ in place of $\mathcal{B}_v(\lambda)$, and $\sigma_u(\ell)$ in place of $\sigma_u(\ell, \lambda)$.

We say that an edge in $\F$ between nodes in the same block is \emph{red}, while inter-block edges are \emph{blue}.
We assign weight $0$ to red edges, and weight $1$ to blue edges.
See \Cref{fig:tree} for an example.

\noindent Our data structure stores:
\begin{itemize}
    \item A top tree representing $\F$.
    In addition to the operations discussed in \Cref{sec:paths}, a top tree also supports lowest common ancestor (LCA) queries, i.e., given two nodes $u, v$ of $\F$, it either reports that $u$ and $v$ belong to different trees in $\F$, or it answers with the deepest vertex $w$ in the unique tree $T$ containing both $u$ and $v$ such that $w$ is an ancestor of both $u$ and $v$.  An LCA query requires $O(\log \eta)$ worst-case time, where $\eta$ is the number of nodes in $\F$.
    
    \item A dictionary $D_v$ for each non-root vertex $v \in V(F)$ that stores a key for each label in $\lambda(e_v)$ and that supports insertions, deletions and predecessor/successor queries. Each key $\ell$ stores, as satellite data, a pointer to node $(\ell, v)$ in $\F$.    
    
    \item A dictionary for each block $\mathcal{B}_v$, where $v$ is a non-leaf vertex in $F$. Such a dictionary stores all elements in $\mathcal{B}_v$ and supports insertions, deletions and predecessor/successor queries w.r.t.\ our order relation on the nodes. 
    
    \item The depth $d(v)$ of each vertex $v$ in the (unique) rooted tree $T$ containing $v$ in $F$, i.e., the hop distance between $v$ and the root of $T$.
\end{itemize}

\subparagraph*{Adapting FixParent.}
FixParent($\ell, v$) takes a non-root vertex $v \in V(F)$ and a label $\ell \in \lambda(e_v)$, and ensures that the edge from node $(\ell, v)$ to its parent in $\F$, if any, is properly set. 

To implement FixParent($\ell, v$) we first \emph{cut} the current edge from $(\ell, v)$ to its parent in $\F$ (if any).
Then we compute $\sigma_v(\ell)$ in $O(\log L)$ worst-case time by searching for the needed successors in $\mathcal{B}_{p(v)}$ and $D_{p(v)}$ (if $p(v)$ exists), according to the definition of $\sigma_v$. Notice that $\sigma_v(\ell)$ might be undefined. 
Finally, if $\sigma_v(\ell)$ is defined, we \emph{link} vertex $(\ell, v)$ with $\sigma_v(\ell)$.
Since link and cut operations require time $O(\log M)$, the overall worst-case time spent by FixParent is $O(\log M)$.

\subparagraph*{Label addition.}
To add label $\ell$ on edge $e_v$, we first insert node $(\ell, v)$ into $\F$, $\ell$ into $D_v$, and $(\ell, v)$ into $\mathcal{B}_{p(v)}$. 
Then, if $v$ is not a leaf in $F$ and $(\ell', u) = \pred((\ell, +\infty), \mathcal{B}_v)$ exists, we perform FixParent($\ell', u$).\footnote{Notice that the set of nodes in $B_v$ does not change following the label addition.} Next, we perform FixParent($\ell, v$).
Finally, if the strict predecessor $(\ell^-, v^-)$ of $(\ell, v)$ in $\mathcal{B}_{p(v)}$ exists, we perform FixParent($\ell^-, v^-$).\footnote{Notice  that  the insertion of  $(\ell, v)$ into $B_{p(v)}$ does not  affect the value (nor the existence) of the strict predecessor of $(\ell, v)$ in $B_{p(v)}$.}

\subparagraph*{Label deletion.}
To delete label $\ell$ from edge $e_v$, we remove the node $(\ell, v)$ from $\F$ along with all its incident edges, we delete $\ell$ from $D_v$, and we delete $(\ell, v)$ from $\mathcal{B}_{p(v)}$.
Then, if $v$ is not a leaf in $F$ and $(\ell', u) = \pred((\ell, +\infty), \mathcal{B}_v)$ exists, we perform FixParent($\ell', u$).
Finally, if the strict predecessor $(\ell^-, v^-)$ of $(\ell, v)$ in $\mathcal{B}_{p(v)}$ exists, we perform FixParent($\ell^-, v^-$). \medskip

\noindent To prove the correctness of our label addition deletion procedures we  need the following lemma which captures the changes to $\F$ following an update:
\begin{lemma}
\label{lm:correct_updates}
    Consider a forest $F$, two functions $\lambda', \lambda : E(F) \to \mathscr{P}(\mathbb{Z})$ and an edge $e_v \in E(F)$ such that, for each $e' \in E(F) \setminus \{e_v\}$, $\lambda(e') = \lambda'(e')$ and $\lambda(e_v) = \lambda'(e_v) \setminus \{ \ell \}$ with $\ell \in \lambda'(e_v)$.
    Let $\F' = \F(\lambda')$ and $\F = \F(\lambda)$.
    Let $U$ be the set containing $(\ell, v)$, $(\ell', u) = \pred((\ell, +\infty), \mathcal{B}_v(\lambda))$ (if it exists), and the strict predecessor $(\ell^-, v^-)$ of $(\ell, v)$ in $\mathcal{B}_{p(v)}(\lambda)$ (if it exists).

    We have that $V(\F) = V(\F') \setminus \{ (\ell, v) \}$ and that all nodes in $V(\F) \setminus U$ have the same parent (or lack of thereof) in both $F$ and $\F'$.
\end{lemma}
\begin{proof}
    We have $V(\F')=V(\F)\setminus\{(\ell,v)\}$ by construction.
    Let $\mathcal{B}=\mathcal{B}(\lambda)$ and $\mathcal{B}'=\mathcal{B}(\lambda')$. To ease the notation we will write $\sigma_x(\ell'')$ in place of $\sigma_x(\ell'', \lambda)$, and $\sigma'_x(\ell'')$ in place of $\sigma_x(\ell'', \lambda')$.    

    We argue that, for the nodes in  $V(F) \setminus U$, $\sigma_x(\ell'') = \sigma'_x(\ell'')$. 
        
    Given a generic node $(\ell'', x) \in V(F) \setminus U$,  the definition $\sigma'_x(\ell'')$ (resp.\ $\sigma_x(\ell'')$) depends only on $\lambda(e_x)$ and on the nodes in $\mathcal{B}'_{p(x)}$ (resp. $\mathcal{B}_{p(x)}$).
    Since $\mathcal{B}'_y = \mathcal{B}_y$ for all $y \neq p(v)$, this implies that we only need to consider 
    the nodes $(\ell'', x)$ in $(\mathcal{B}_v \cup \mathcal{B}_{p(v)}) \setminus U$. 
    
    We analyze the cases $(\ell'', x) \in  \mathcal{B}_v \setminus U$ and $(\ell'', x) \in  \mathcal{B}_{p(v)} \setminus U$ separately.

    We consider  $(\ell'', x) \in  \mathcal{B}_v \setminus U$ first.
    In this case we have $p(x)=v$.
    Since $\mathcal{B}_v = \mathcal{B}'_v$,
    the strict successor of $(\ell'', x)$ is the same in both sets.
    Suppose then that $\succ(\ell'', \lambda(e_v)) \neq \succ(\ell'', \lambda'(e_v))$.
    We must have $\ell'' \le \ell$ which implies the existence of $(\ell', u)$ and that $(\ell'', x)$ precedes $(\ell',u)$ in $\mathcal{B}_v = \mathcal{B}'_v$. Then $\ell'' \le \ell'$ and 
    $\sigma_x(\ell'')$ and $\sigma'_x(\ell'')$ are both defined as the strict successor of $(\ell'', x)$ in $\mathcal{B}_v = \mathcal{B}'_v$.
    
    We now consider $(\ell'', x) \in \mathcal{B}_{p(v)} \setminus U$. Let $y = p(x) = p(v)$.
    Clearly, $\succ(\ell'', \lambda(e_y)) = \succ(\ell'', \lambda'(e_y))$ since $e_y \neq e_v$. 
    Moreover, the only node that can have a different strict successor in $\mathcal{B}_{y}$ and $\mathcal{B'}_{y}$
    is the strict predecessor $(\ell^-, v^-)$ of $(\ell, v)$ (which coincides in $\mathcal{B}_{y}$ and $\mathcal{B'}_{y}$).
    Since $(\ell^-, v^-) \in U$  $(\ell'', x) \neq (\ell^-, v^-)$, hence $\sigma_x(\ell'') = \sigma'_x(\ell'')$.
\end{proof}

Then, the correctness of the label addition procedure follows from \Cref{lm:correct_updates} with $\lambda$ (resp.\ $\lambda'$) chosen as the function that maps each edge to its labels before (resp.\ after) the addition, once one observe that FixParent is invoked on all vertices of the set $U$ defined by the lemma. Symmetrically, the correctness of the label deletion procedure follows from \Cref{lm:correct_updates} when the roles of $\lambda$ and $\lambda'$ are reversed.

The overall worst-case time required per update is $O(\log M)$ since we only perform a constant number of (strict) predecessor lookups, edge/node deletions from $\F$, and calls to FixParent. Observe that \Cref{lm:correct_updates} implies that, when node $(\ell ,v)$ is deleted, its degree in $\F$ is constant (since each child of $(\ell, v)$ changes its parent following the deletion).

\subparagraph*{Answering upward EA queries and downward LD queries.}
We first discuss how to report $\EA(u,v, t)$ when $u$ is a proper descendent of $v$ in $F$.\footnote{Since the topology of $F$ does not change, checking whether $u$ is a proper descendent of $v$ can be done in constant time after $O(n)$-time preprocessing, where $n$ is the number of vertices in $F$.}
To do so, we start by finding $\ell = \succ(t, \lambda(e_u))$. If $\ell = +\infty$, we answer with $+\infty$, otherwise we query $\F$ for the $(d_v-d_u-1)$-th weighted level ancestor of $(\ell,u)$. If such an ancestor $(\ell^*, w)$ exists we answer with $\ell^*$, otherwise we answer with $+\infty$. This requires $O(\log M)$ worst-case time since it only involves a successor lookup in $D_u$ and a weighted level ancestor query on the top tree representing $\F$.

The correctness of our query procedure stems by a structural property of $\F$ captured by the following lemma, whose proof is given in \Cref{apx:ea_query_correctness}.

\begin{restatable}{lemma}{correctnessEAquery}
    \label{lemma:correctness_ea_query}
    Let $u, v \in V(F)$ where $v$ is a proper ancestor of $v$, let $t \in \mathbb{Z} \cup \{-\infty\}$, and define $\ell = \succ(t, \lambda(e_u))$. If (i) $\ell = +\infty$ or (ii) $\ell < +\infty$ and $\LA( (\ell, u), d_v - d_u - 1)$ does not exist, then $\EA(u, v, t) = +\infty$. Otherwise, $\LA( (\ell, u), d_v - d_u - 1) = (\ell^*, w)$ where $w$ is the unique ancestor of $u$ such that $p(w)=v$ and $\ell^* = \EA(u, v, t)$.
\end{restatable}

To answer downward LD queries, we maintain a mirrored data structure for $-F$, in which each original label $\ell$ is replaced with $-\ell$.
\Cref{lemma:LD_eq_minusLA} allows us to answer downward LD queries in $O(\log M)$ worst-case time by performing upward EA queries on the data structure for $-F$.

\subparagraph*{Answering upward LD queries and downward EA queries.}
To report $\LD(u, v, t)$ when $u$ is a proper descendent of $v$ in $F$, we binary search for the largest label $\ell^* \in \lambda(e_u)$ such that $\EA(u, v, \ell^*) \le t$. Then, we report $\ell^*$. The correctness of our query immediately follows from the fact that the values  $\EA(u, v, \ell)$ are monotonically non decreasing w.r.t.\ $\ell$. 
This requires a worst-case time of $O(\log L \cdot \log M)$. To report $\EA(w, v, t)$ when $v$ is a proper descendant of $w$ in $F$, we compute $\LD(v,w,-t)$ on the data structure for $-F$ and answer with $-\LD(v,w,-t)$.
The worst-case time required is $O(\log L \cdot \log M)$.

\subparagraph*{Answering general EA and LD queries.}
To answer a general $\EA(u,v,t)$ query we compute the lowest common ancestor $w$ of $u$ and $v$ in $F$ in constant time using the data structure in \cite{BenderF00}, and we return $\EA(w, v, \EA(u,v, t))$ where $\EA(w,v, \cdot)$ is a downward query, and $\EA(u, v, \cdot)$ is an upward query.
General LD queries can similarly be answered by using the data structure for $-F$. The worst-case time required to answer such queries is $O(\log L \cdot \log M)$.

\subparagraph*{Answering temporal reachability queries.}
To report whether a vertex $v$ is reachable from a vertex $u$ using a temporal path in $F$ that departs no earlier than $t_d$ and arrives no later $t_a$, we compute the lowest common ancestors $w$ of $u$ and $v$ in $F$ in constant time and we answer affirmatively iff $\EA(u,w,t_d) \le \LD(w, v, t_a)$. 
The overall worst-case time required is $O(\log M)$ since we only need to perform an upward EA query and a downward LD query.

\subsection{Supporting link and cut operations}
To support general \emph{link} and \emph{cut} operations we additionally maintain $F$ using a top tree.\footnote{Some technical care is needed to obtain the stated bounds, which only depend on $M$, when $M = o(n)$. This can be achieved by not actually storing singleton vertices in the top tree, so that the number of vertices in the top tree is always $O(M)$. The operations of our data structure are easy to adapt to handle this edge case. E.g., whenever a link operation on $F$ involves some singleton vertex $v$, we can add $v$ to the top tree immediately before performing the link.} Each time that an edge is added/removed from $F$, we perform the corresponding operation on the backing top tree.
Moreover, we no longer explicitly maintain the depths $d_v$, but rather we query the top tree of $F$ every time any such depth is needed.\footnote{Indeed, the top tree can report the root of the tree in $F$ that contains $v$ and the hop-distance between any two nodes in $F$ in logarithmic time.} Similarly, all LCA queries on $F$ are now performed using the backing top tree.

Insertions and deletion of singleton vertices are straightforward, therefore we only discuss how to handle link and cut operations. 

\subparagraph*{Link operations.}
To implement a link operation where $u$ is the root of some tree $T$ in $F$ and $v$ is some vertex of a tree $T'$ of $F$ with $T' \neq T$, and a label $\ell \in \mathbb{Z}$, we first \emph{link} $u$ and $v$ in $F$ by adding edge $(u,v)$, so that the parent of $u$ becomes $v$, and we add label $\ell$ to $(u,v)$ as explained in \Cref{subsec:fixed_forest}.
The worst-case time required is $O(\log M)$.

\subparagraph*{Cut operations.}
To cut an edge $(u,v)$ of $F$ with a single label $\ell \in \lambda( (u,v) )$, where $u$ is a child of $v$ w.l.o.g., we first remove the only label $\ell$  (corresponding to the only key in $D_u$) from $(u,v)$ as explained in \Cref{subsec:fixed_forest}. Then, we \emph{cut} $(u,v)$ from $F$, thus creating a new tree rooted in $u$.
The worst-case time required is $O(\log M)$.

\section{Our data structure for temporal forests with latencies}\label{sec:temporal_forests_with_latencies}

Our model of temporal graphs can be generalized by introducing \emph{latencies}.
In \emph{temporal graphs with latencies}, each edge $e$ is associated with a collection of pairs $(\ell, d)$ encoding that edge $e$ can be traversed at the \emph{departure time} $\ell$ stating from one of its endvertices in order to reach the other endvertex at time $\ell + d$. The value $d$ is called a \emph{latency}.

Equivalently, each (activation time, latency) pair $(\ell, d)$ can be expressed as $(\ell, \alpha)$, where $\alpha = \ell + d$ is the \emph{arrival time}. These two representations are clearly equivalent, but the latter one results in a lighter notation for our purposes. Therefore, in the rest of the section we will adopt the (departure time, arrival time) convention and we
accordingly define $\lambda(e)$ as the set of all (departure time, arrival time) pairs $(\ell, \alpha)$ associated with edge $e$.\footnote{The special case in which all labels $(\ell, \alpha)$ have $\alpha=\ell$ corresponds to the model used in the previous sections.}

A \emph{temporal path} $\pi$ from vertex $u\in V(F)$ to vertex $v \in V(F) \setminus \{u\}$ in $F$ is a sequence of triples $\langle (e_1, \ell_1, \alpha_1),  (e_2, \ell_2, \alpha_2), \dots, (e_k, \ell_k, \alpha_k) \rangle$ such that $\langle e_1, e_2, \dots, e_k \rangle$ is a path from $u$ to $v$  in $F$, $(\ell_i, \alpha_i) \in \lambda(e_i)$ for all $i=1,\dots,k$, and $\alpha_i \le \ell_{i+1}$ for all $i=1, \dots, k-1$.
The departure time of $\pi$ is $\ell_1$ and its arrival time is $\alpha_k$. The notions of earliest arrival paths, latest departure paths,  reachability, and the corresponding queries extend naturally.

In this section we argue that our data structure for temporal forests can be adapted to additionally support latencies. This comes at the cost of turning our worst-case bounds on the time complexities of the link and cut operations into amortized bounds that hold for the incremental case, in which only link operations are allowed, and in the decremental case, in which only cut operations are allowed.\footnote{Our data structure can still handle arbitrary sequences of link and cut operations, but the amortized bounds do not hold for this case.}

\begin{theorem}\label{thm:incremental_decremantal_latencies}
    Given a temporal forest of rooted trees with latencies, it is possible to build a data structure of linear size that supports EA and LD queries in $O(\log L \cdot \log M)$ worst-case time per operation, and reachability queries in $O(\log M)$ worst-case time per operation, where $L$ is the maximum number of labels on the same temporal edge, and $M$ is the total number of (not necessarily distinct) labels in the forest at the time of the operation.
    In the incremental case, the data structure also supports insertions of singleton vertices, label insertions, and link operations in amortized time $O(\log M)$.
    In the decremental case, the data structure also supports deletions of singleton vertices, label deletions, and cut operations in amortized $O(\log M)$ time and amortized building time of $O(M\log L)$.\footnote{In the decremental case, we can naturally assume that the initial number of vertices $n$ in $F$ satisfies $n = \Omega(M)$.}
\end{theorem}

As before, we consider a temporal forest with latencies $F$ containing rooted temporal trees and, for a non-root vertex $v$, we define $e_v$ as the edge from $v$ to its parent $p(v)$ in the unique tree of $F$ containing $v$.
We define $\F(\lambda)$ as the forest containing a \emph{node} $(\alpha, \ell, v)$ for each non-root vertex $v$ and for each $(\ell, \alpha) \in \lambda(e_v)$.
For each non-leaf vertex $v \in V(F)$, we also define the \emph{block of $v$} as the set $\mathcal{B}_v(\lambda)$ containing all nodes $(\alpha, \ell, u)$ such that $p(u) = v$.
We fix an arbitrary order of the vertices of $V(F)$ and we think of the nodes of $\F$ as being ordered w.r.t.\ the order relation that compares nodes lexicographically, i.e., $(\alpha, \ell, u)$ precedes $(\alpha', \ell', v)$ if (i) $\alpha < \alpha'$, or (ii) $\alpha = \alpha'$ and $\ell < \ell'$, or (iii) $\alpha = \alpha'$, $\ell = \ell'$, and $u$ precedes $v$ in the chosen ordering of the vertices.
Informally, we first order the nodes in a block by arrival time, then by departure time, and finally by their corresponding vertex in $F$.

We now define the analogue of the function $\sigma$ for temporal forests with latencies.
More precisely, given a non root vertex $u$ in $F$ with parent $v = p(u)$ and $(\ell, \alpha) \in \lambda(e_u)$, we let $(\alpha^+, \ell^+, u^+)$ be the \emph{strict successor} of $(\alpha, \ell, u)$ in $\mathcal{B}_v(\lambda)$, if any.
If $v$ is a root of a tree in $F$, then $\sigma_{u}(\alpha, \ell, \lambda) = (t^+, \ell^+, u^+)$ if $(t^+, \ell^+, u^+)$ exists, otherwise $\sigma_{u}(\alpha, \ell, \lambda)$ is undefined.
When $v$ is not a root of any tree in $F$, we define NextHop($\alpha, \ell, u$) as 
the node $(\alpha',\ell', v)$ such that $(\ell', \alpha') \in \lambda(e_v)$, $\ell' \geq \alpha$, and $\alpha'$ is minimized, breaking ties in favor of labels with the largest departure time (such a node might not exists).

We observe that if NextHop$(\alpha^+,\ell^+,u^+)$ exists, if it exists, either coincides with $(\alpha',\ell', v)$ or it follows $(\alpha', \ell', v)$ in $\mathcal{B}_{p(v)}$ w.r.t.\ our ordering. 
Thus, if both  $(\alpha^+, \ell^+, u^+)$  and $(\alpha',\ell', v)$ exist, we define:
$\sigma_u(\alpha, \ell, \lambda) =
\begin{cases}
   (\alpha^+, \ell^+, u^+) & \text{if } \alpha^+ \le \ell'; \\
   (\alpha', \ell', v) & \text{if } \ell' < \alpha^+.
\end{cases}$ 

\noindent Here the condition $\alpha^+ \le \ell'$ is equivalent to the following: $\text{NextHop}(\alpha^+, \ell^+, u^+)$ exists and coincides with $(\alpha', \ell', v)$. If neither  $(\alpha^+, \ell^+, u^+)$  nor $(\alpha',\ell', v)$ exist, then $\sigma_u(\alpha, \ell, \lambda)$ is undefined and thus $(\alpha, \ell, \lambda)$ is a root.
Otherwise, $\sigma_u(\alpha, \ell, \lambda)$ is defined as the only node that exists among $(\alpha^+, \ell^+, u^+)$ and $(\alpha',\ell', v)$.
As usual, we drop the parameter $\lambda$ from $\F(\lambda)$, $\sigma_u(\alpha, \ell, \lambda)$, and $\mathcal{B}(\lambda)$ whenever $\lambda$ is clear from context.

\begin{figure}
    \centering
    \includegraphics[width=\textwidth]{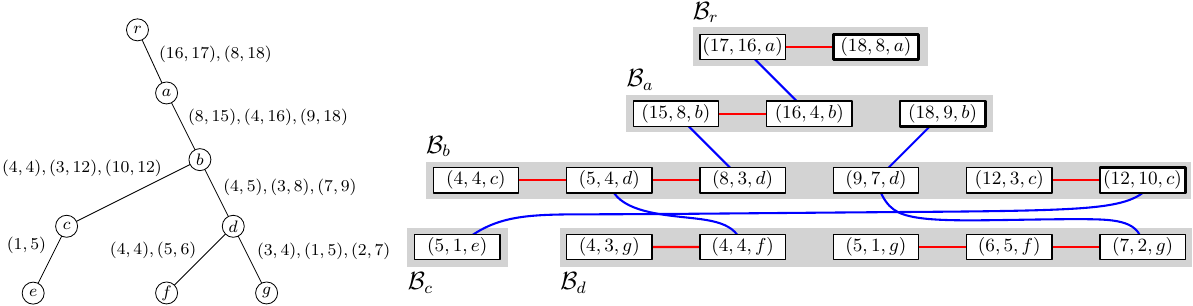}
    \caption{Left: a sample rooted temporal tree with latencies. Right: a representation of the forest $\F$ maintained by the data structure for temporal forests with latencies of \Cref{sec:temporal_forests_with_latencies}.}
    \label{fig:tree_latency}
\end{figure}

\noindent Our data structure is analogous to that of \Cref{sec:temporal_forests}, with the following exceptions:
\begin{itemize}
    \item each dictionary $D_v$ storing the pairs in $\lambda(e_v)$ now supports queries of the following form: given a range of values of interest for $\ell$ (resp.\ $\alpha$), return the minimum/maximum value of $\alpha$ (resp.\ $\ell$) w.r.t.\ all labels $(\ell, \alpha) \in \lambda(e_v)$ such that $\ell$ (resp.\ $\alpha$) is in the sought range (notice that value might not exist). This can be done in time $O(\log |\lambda(e_v)|)$ using, e.g., priority search trees \cite{McCreight85}, which require space $O(|\lambda(e_v)|)$.

    \item for each non-leaf node $v$ we say that a node $(\alpha, \ell, u) \in \mathcal{B}_v$ is a \emph{head} of $\mathcal{B}_v$ if $(\alpha, \ell, u)$ either has no parent in $\F$ or it is linked to its parent with a blue edge. We store a dictionary $\mathcal{H}_v$ that contains all heads of $\mathcal{B}_v$.
\end{itemize}

\noindent Since queries are analogous to the latency-free case, we only focus on label additions/deletions. 

\subparagraph*{An auxiliary procedure.}
The auxiliary procedure FixParent($\ell, \alpha, v$) for the case with latencies is similar to FixParent in the case without latencies, since it only uses the definition of $\sigma_v$ as before. The only difference is that that the execution of FixParent also needs to update $\mathcal{H}_{p(v)}$ taking into account the new parent of $(\ell, \alpha, v)$.

\subparagraph*{Label addition.}
To add label $(\ell, \alpha)$ on edge $e_v$, we first insert node $(\alpha, \ell, v)$ into $\F$, $(\ell, \alpha)$ into $D_v$, and $(\alpha, \ell, v)$ into $\mathcal{B}_{p(v)}$.
We find the maximum value $\ell^-$ of $\ell''$ among all pairs $(\ell'', \alpha'') \in \lambda(e_v)$ with $\alpha'' < \alpha$ using $D_v$.

Let $(\alpha_1, \ell_1, u_1), \dots, (\alpha_k, \ell_k, u_k)$ be
all the nodes in $\mathcal{B}_v$ such that $\ell^- < \alpha_i \le \ell$ sorted w.r.t.\ our order (see \Cref{fig:addition_latencies}). We observe that all these nodes are consecutive in $\mathcal{B}_v$. 
The parent of $(\alpha_k, \ell_k, u_k)$ becomes $(\alpha, \ell, v)$, while the parent of all $(\alpha_i, \ell_i, u_i)$ becomes $(\alpha_{i+1}, \ell_{i+1}, u_{i+1})$.
This requires updating all heads in $\mathcal{H}_v$ between $(\alpha_1, \ell_1, u_1)$ and  $(\alpha_{k-1}, \ell_{k-1}, u_{k-1})$, and can be in time $O(h \log M)$, where $h$ is the number of such heads.
This causes all the updated heads to be removed from $\mathcal{H}_v$, and $(\alpha_k, \ell_k, u_k)$ to become a new head (if that was not already the case). Next, we perform FixParent($\alpha, \ell, v$).
Finally, if the strict predecessor $(\alpha^*, \ell^*, v^*)$ of $(\alpha, \ell, v)$ in $\mathcal{B}_{p(v)}$ exists, we perform FixParent($\alpha^*, \ell^*, v^*$).

We now argue that the amortized time complexity of each label insertion is $O(\log M)$ in the incremental case using the accounting method. 
Let $c>0$ be a sufficiently large constant. We keep a coin of value at least $c \cdot H_M$, where $H_i = \sum_{j=1}^i \frac{1}{j}$ denotes the $i$-th harmonic number, on each node of $\F$ that is either a root or is linked to its parent with a blue edge.
When a new label is inserted, we pay up to $c M \cdot \frac{1}{M+1}$ for the increase in value of the existing coins so that each coin has value $c H_{M+1}$, and $c H_{M+1}$ for the coin on the new node $(\alpha, \ell, v)$ in $\F$.
We also add a coin of value $c H_{M+1}$ on each of  $(\alpha, \ell, v)$ and $(\alpha_k, \ell_k, u_k)$. 
Notice that each node of $(\alpha_1, \ell_1, u_1), \dots, (\alpha_{k-1}, \ell_{k-1}, u_{k-1})$ that changes the edge towards its parent either had no parent, or it was linked to its previous parent with a blue edge. Moreover, the new edges towards its parent must be red. This means that such a node had a coin of value $c H_{M+1}$ that we can use to pay for the cost of the rewiring.
All the other operations performed during a label addition cost $O(\log M)$ worst-case time. 

\subparagraph*{Label deletion.}
To delete the label $(\ell, \alpha)$ from edge $e_v$, we remove the node $(\alpha, \ell, v)$ from $\F$ along with all its incident edges, we delete $(\ell, \alpha)$ from $D_v$, and we delete $(\alpha, \ell, v)$ from $\mathcal{B}_{p(v)}$, and possibly from $\mathcal{H}_{p(v)}$.

If $v$ is not a leaf in $F$, let $(\alpha_1, \ell_1, u_1), \dots, (\alpha_k, \ell_k, u_k)$ be the nodes, in order, that had the NextHop equal to $(\alpha,\ell,v)$. Notice that such nodes induce a red path in $\F$, and that $(\alpha_k, \ell_k, u_k)$ was the unique blue child of $(\alpha, \ell, v)$ before the deletion.  

As a consequence of the deletion, some of the nodes in $(\alpha_1, \ell_1, u_1), \dots, (\alpha_k, \ell_k, u_k)$ will become heads of $\mathcal{B}_v$, and hence will have a blue edge towards their new parent in $\mathcal{B}_{p(v)}$. We find $(\alpha_k,\ell_k,u_k)$ in $O(\log M)$ worst-case time via binary search in $\mathcal{B}_v$ and we use it discover such heads as follows: 
Initially $(\alpha^*, \ell^*, u^*) = (\alpha_k, \ell_k, u_k)$, and  $z = \text{NextHop}(\alpha_k, \ell_k, u_k)$. We binary search $\mathcal{B}_v$ for the rightmost node $(\alpha', \ell', u')$ (w.r.t.\ our ordering) that precedes $(\alpha_k, \ell_k, u_k)$ and is such that $\text{NextHop}(\alpha', \ell', u') \neq z$. We mark node $(\alpha', \ell', u')$ as a head, and repeat the above procedure using $(\alpha^*, \ell^*, u^*) = (\alpha', \ell', u')$, and  $z = \text{NextHop}(\alpha', \ell', u')$.
We stop the above procedure as soon as $z$ precedes $(\alpha, \ell, v)$ in $B_{p(v)}$.
Then this requires $O(\log M)$ time, plus and additional $O(\log M)$ time per discovered head since we can check whether $\text{NextHop}(\alpha', \ell', u') \neq z$ in constant time. Indeed, calling $z = (\alpha_z, \ell_z, v)$, we can define $\ell_z^-$ as the largest value of $\ell''$ among the pairs $(\ell'', \alpha'') \in \lambda(e_v)$ with $\alpha'' < \alpha_z$ (such $\ell_z^-$ can be found by querying $D_v$). Then $\text{NextHop}(\alpha', \ell', u') \neq z$ iff $\ell' < \ell_z^-$.

We run FixParent on $(\alpha_k, \ell_k, u_k)$ and on all the marked heads, which also updates $\mathcal{H}_v$.

Finally, if the strict predecessor $(\alpha^-, \ell^-, v^-)$ of $(\alpha, \ell, v)$ in $\mathcal{B}_{p(v)}$ exists, we perform FixParent($\alpha^-, \ell^-, v^-$).

We now argue that the amortized time complexity of each label deletion is $O(\log M)$ in the decremental case using the accounting method. 
Let $c$ be sufficiently large constant. We keep a coin of value at least $c \cdot H_M$ on each node of $\F$ that is either a root or is linked to its parent with a red edge. Hence, we pay an amortized cost of $O(M \log M)$ at construction time. When a label is deleted, we pay $2 c H_M$ to add a coin on each of  $(\alpha^-, \ell^-, v^-)$, and $(\alpha_k, \ell_k, u_k)$. Moreover, since each marked head $(\alpha', \ell', u')$ had a red edge towards its parent before the deletion, we spend such a coin to pay for the execution of FixParent on $(\alpha', \ell', u')$.
All the other operations performed during a label deletion cost $O(\log M)$ worst-case time.

\begin{figure}[t]
    \includegraphics[width=.95\textwidth]{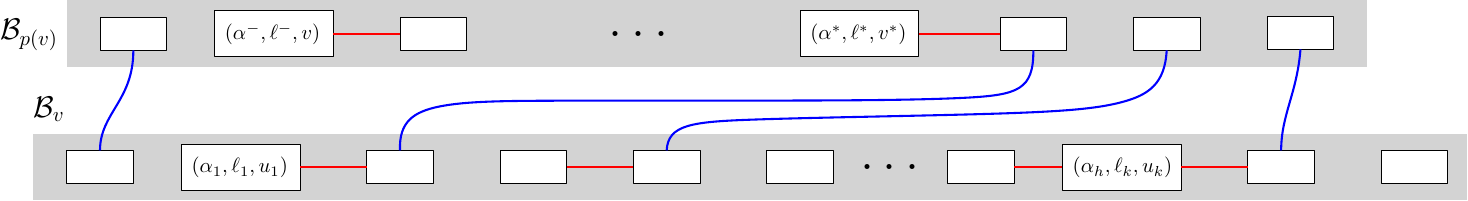}\\[.35cm]
    \includegraphics[width=.95\textwidth]{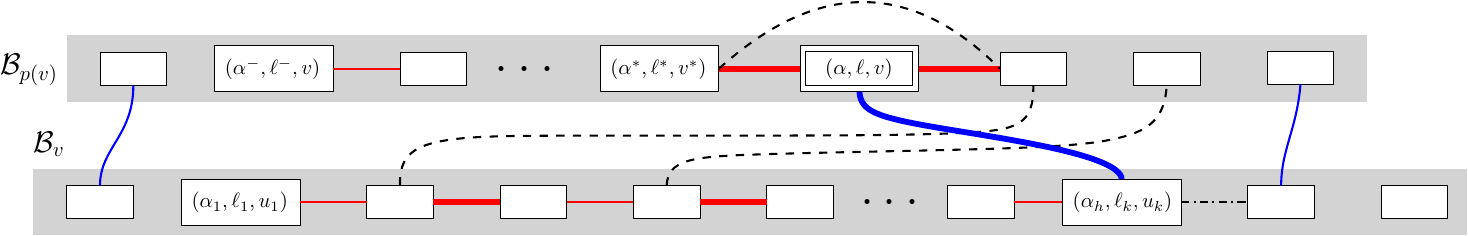}
    \caption{A qualitative representation of the changes resulting from the addition of label $(\ell, \alpha)$ on edge $e_v$.
    On the top: the blocks $\mathcal{B}_{p(v)}$ and $\mathcal{B}_{v}$ before inserting node $(\alpha, \ell ,v)$.
    On the bottom: the new state of $\mathcal{B}_{p(v)}$ and $\mathcal{B}_{v}$.
    Bold lines represent new edges, while dashed lines represent removed ones.}
    \label{fig:addition_latencies}
\end{figure}

\subparagraph*{Uniform latencies.}
A special case of temporal graphs with latencies is the one with \emph{uniform latencies}, where all time labels have the same latency.
In this case, the ordering defined on the nodes of $\F$ is exactly the same as the one used in \Cref{sec:temporal_forests}, for the case without latencies.
For this reason, the data structure we presented in this section guarantees a worst-case update time of $O(\log M)$ because only a constant number of nodes in $\F$ can change their parents.

\begin{corollary}\label{cor:uniform_latencies}
    Given a temporal forest of rooted trees with uniform latencies, it is possible to build in $O(n + M\log L)$ time, a data structure of linear size, $n$ is the number of vertices,  $L$ is the maximum number of labels on the same temporal edge, and $M$ is the total number of (not necessarily distinct) labels in the forest at the time of the operation.
    The data structure supports label additions/deletions, link/cut operations and addition/deletion of singleton vertices in $O(\log M)$ worst-case time per operation, EA and LD queries in $O(\log L \cdot \log M)$ worst-case time per operation, and reachability queries in $O(\log M)$ worst-case time per operation.
\end{corollary}

\bibliographystyle{plain}
\bibliography{references}

\clearpage

\clearpage
\appendix

\section{A data structure using the heavy-light decomposition}
\label{apx:heavy_light_decomposition}

Given a simple path $P$ and two vertices $x,y$ in $P$, we denote with $P[x:y]$ the subpath of $P$ that has $x$ and $y$ as endvertices.
Moreover, given two paths $P = \langle x_1, \dots, x_h\rangle$ and $P' = \langle y_1, \dots, y_k\rangle$ with $x_h = y_1$, we denote with $P \circ P'$ the concatenation of $P$ and $P'$, i.e., the (not necessarily simple) path $\langle x_1, \dots, x_h, y_2, \dots, y_k \rangle$.

Consider a temporal forest $F$ of rooted tress with $n$ vertices in total.
We can compute, in time $O(n)$, a collection $\mathcal{P} = \{P_1, P_2, \dots\}$ of ancestor-descendant paths in $F$ such that (i) each edge of $F$ belongs to exactly one path, and (ii) given two vertices $u$, $v$ of a tree $T$ in $F$ such that $v$ is a proper ancestor of $u$ or vice-versa, we can find, in $O(\log n)$ time, $k = O(\log n)$ vertices $u=x_1, \dots, x_k=v$ and $k-1$ paths $P_1, \dots, P_{k-1}$ such that $P_i$ contains  $x_{i}$ and $x_{i+1}$ and the unique path from $u$ to $v$ in $F$ can be written as $P_1[x_1 : x_2] \circ P_2[x_2 : x_3] \circ \dots \circ P_k[x_{k-1} : x_k]$. 

Such a decomposition $\mathcal{P}$ can be found by extending each path of a heavy-light decomposition \cite{SLEATOR1983362} of each tree $T$ in $F$ by one additional edge towards its root.
The heavy-light decomposition of a rooted tree $T$ is defined recursively as follows: let $P$ be a path in $T$ obtained by starting from the root and iteratively traversing an edge $(u,v)$ from the current vertex $u$ to a child $v$ of $u$ that maximizes the number of vertices in the subtree of $T$ rooted at $v$. The decomposition contains $P$ and all paths resulting from applying the above procedure to all (rooted) trees of the forest obtained from $T$ by deleting all vertices in $P$.

Our data structure for temporal trees stores $\mathcal{P}$ and the data structure $\mathcal{O}_i$ of \Cref{thm:static_path} for each path $P_i \in \mathcal{P}$.
Moreover, for each edge $e$ of $F$ we store a reference to the path $P_i$ containing $e$.
Finally, we store a linear-size oracle $\mathcal{O}_{LCA}$ capable of answering \emph{lowest common ancestor} (LCA) queries on $T$ in constant time \cite{BenderF00}.
Since the size of $\mathcal{O}_i$ is linear in the overall number of vertices and edge labels in $P_i$, the size of our data structure is $O(n + M)$.

\subparagraph{Handling label additions and deletions}

To handle the addition/deletion of a label $\ell$ to/from edge $e$, we find the path $P_i \in \mathcal{P}$  containing $e$ and we delegate the update operation to $\mathcal{O}_i$.

\subparagraph{Handling earliest arrival and latest departure queries}

To answer a query $\EA(u, v, t)$ when $u$ is an ancestor of $v$ or vice-versa, we can use property (ii) above to find a set of vertices $u=x_1, x_2, \dots, x_k=v$ that lie on the unique path $\pi$ between $u$ and $v$ in $T$, in this order, and a set of paths $P_1, \dots, P_{k-1} \in \mathcal{P}$ such that $\pi[x_i, x_{i+1}] = P_i[x_i, x_{i+1}]$.

We iteratively compute $t_i = \EA(u, x_i, t)$ in increasing order of $i$ until we find $t_k = \EA(u, v, t)$. Clearly $t_1 = t$. To compute $\EA(u, x_i, t)$ for $i>1$ we use the identity of \Cref{lemma:LD_eq_minusLA}:
\[
    \EA(u, x_i, t) = \EA(x_{i-1}, x_i, \EA(u, x_{i-1}, t)) = \EA(x_{i-1}, x_i, t_{i-1}).
\]
Since $x_{i-1}$ and $x_{i}$ belong to the same path $P_{i-1} \in \mathcal{P}$, we can find $\EA(x_{i-1}, x_i, t_{i-1})$ in time $O(\log M)$ by querying $\mathcal{O}_{i-1}$.

The overall time required is $O(k \cdot \log M) = O(\log n \cdot \log M)$.

To answer a generic query $\EA(u, v, t)$ we observe that any temporal path from $u$ to $v$ in $F$ must traverse the lowest common ancestor $w$ between $u$ and $v$ in $F$ and hence $\EA(u, v, t) = \EA(w, v, \EA(u, w, t))$. 
We can hence find $w$ in constant time using $\mathcal{O}_{LCA}$ and recast the query for $\EA(u, v, t)$ as two consecutive sub-queries: the first sub-query asks the earliest arrival time $t'$ between $u$ and $w$ in $T$, while the second subquery computes $\EA(w, v, t')$. Crucially, both $u$ and $v$ are descendants of $w$, hence the previous strategy can be used twice to find $\EA(u, v, t)$ in time $O(\log n \cdot \log M)$.

\clearpage

\section{Correctness of our upward EA query procedure}
\label{apx:ea_query_correctness}

This appendix is devoted to proving \Cref{lemma:correctness_ea_query}. We start with a preliminary lemma.

\begin{lemma}
    \label{lemma:path_query_correctness_1}
    Let $u, v \in V(F)$ with $v = p(u)$ (resp.\ $v=p(p(u))$), let $t \in \mathbb{Z} \cup \{-\infty\}$, and define $\ell = \succ(t, \lambda(e_u))$. 
    If (i) $\ell = +\infty$ or (ii) $\ell < +\infty$ and $\LA( (\ell, u), d_v - d_u - 1)$ does not exist, then $\EA(u, v, t) = +\infty$. Otherwise, $\LA( (\ell, u), d_v - d_u - 1) = (\ell^*, w)$ where $w=u$ (resp.\ $w=p(u)$).
\end{lemma}
\begin{proof}
    The claim is trivially true when $\ell = +\infty$, hence in the rest of the proof we only consider $\ell < +\infty$.

    If $v=p(u)$ we have $\LA( (\ell, u), d_v-d_u-1 ) = \LA( (\ell, u), 0 ) = (\ell, u)$ and $\EA(u, v, t) = \succ(t, \lambda(e_u)) = \ell$.

    If $v = p(p(u))$, consider the longest path $\pi$ from $(\ell, u)$ to one of its ancestors in $\F$ that consists of only red edges.
    Let $(\ell_0, x_0) = (\ell, u), (\ell_1, x_1), \dots, (\ell_k, t_k)$ be the vertices of $\pi$, in order.
    We prove by reverse induction on $i=k, \dots, 0$  that $\EA(x_i, v, \ell_i) = \EA(x_k, v, \ell_k)$.
    The base case $i=k$ trivially holds, therefore we assume that the claim holds for $i+1$ and we prove that it also holds for $i < k$. 

    Let $w = p(u)$.
    We have $\sigma_{x_i}(\ell_i) = (\ell_{i+1}, x_{i+1})$, hence 
    $\ell_i \le \ell_{i+1} \le \succ( \ell_i, \lambda(e_w) )$ and we cab write:
    \[
        \EA(x_i, v, \ell_i) = \succ(\ell_i, \lambda(e_w))
        = \succ(\ell_{i+1}, \lambda(e_w))
        =\EA(x_{i+1}, \ell_{i+1}) = \EA(x_k, v, \ell_k).
    \]

    If $(x_k, \ell_k)$ has no parent in $\F$, then $\LA(u, d_v - d_u - 1) = \LA(u, 1)$ does not exist and $\succ(\ell_k, \lambda(e_w)) = +\infty$, which implies $\EA(u, v, \ell) = \EA(x_k, v, \ell_k) =\succ(\ell_k, \lambda(e_w)) = +\infty$.

    If $(x_k, \ell_k)$ has a parent $(w, \ell')$ in $\F$, then $\LA(u, d_v - d_u - 1) = \LA(u, 1) = (w, \ell')$. By definition of $\sigma_{x_k}$, $\ell' = \succ( \ell_k, \lambda(e_w) ) = \EA(x_k, v, \ell_k)$, hence $\EA(u, v, \ell) = \EA(x_k, v, \ell_k) = \ell'$.
\end{proof}

\noindent We are now ready to prove the following:

\correctnessEAquery*
\begin{proof}
    The proof is by induction on $d_v-d_u$. The base cases are $d_v - d_u = 1$ and $d_v - d_u = 2$ and are proved in \Cref{lemma:path_query_correctness_1}.

    For the inductive step we consider $d_v - d_u > 2$, we assume that the claim holds for all nodes $u',v'$ where $v'$ is a proper ancestor of $u'$ and  $d_{v'} - d_{u'} < d_v - d_u$, and we prove that it also holds for $u$ and $v$.

    Let $x = p(u)$ and $y = p(x)$. Defining $\ell' = \EA(u, y, t)$, 
    we have that either $\ell' = +\infty$ or $\ell' \in \lambda(e_x)$. In any case, $\EA(x, y, \ell')=\ell'$. Let  $\ell'' = \EA(x, v, \ell')$. We have:
    \begin{equation}
    \label{eq:ea_query_correctness_proof}
    \begin{alignedat}{2}
        \EA( u, v, t ) &= \EA(y, v, \EA( u, y, t )) = \EA(y, v, \ell')
        =  \EA(y, v, \EA(x, y, \ell')) \\
        &= \EA(x, v, \ell')
        =  \ell''.
    \end{alignedat}
    \end{equation}

    We first handle the case in which $\LA( (\ell, u), d_v-d_u-1 )$ exists.
    Then $\LA( (\ell, u), 1 )$ also exists and, by induction hypothesis, it is $(\ell' , x)$. We can then write: 
    \[
        \LA( (\ell, u), d_v-d_u-1 ) = \LA( \LA((\ell, u), 1), d_v-d_u-2 ) = \LA( (\ell', x), d_v - d_u -2),
    \]
    which guarantees the existence of $\LA( (\ell', x), d_v - d_u -2)$ and allows us to use the induction hypothesis once again to conclude that  $\LA( (\ell, u), d_v-d_u-1 )  = \LA( (\ell', x), d_v - d_u -2) = (\ell'', w)$, where $w$ is an ancestor of $x$ (and of $u$) and $p(w) = v$.
    
    We now consider the case in which $\LA( (\ell, u), d_v-d_u-1 )$ does not exist. Here the goal is to show that $\EA( u, v, t ) = +\infty$. In this case, the above arguments imply that either (i) $\LA( (\ell, u), 1 )$ does not exist, or (ii) $\LA( (\ell, u), 1 ) = (\ell', x)$  and $\LA( (\ell', x), d_v-d_u-2 )$ does not exist.
    
    If $\LA( (\ell, u), 1 )$ does not exist then, by induction hypothesis, $\EA( u, y, t) = +\infty$ and  $\EA( u, v, t ) = \EA(y, v, \EA( u, y, t )) = \EA(y,v,+\infty) = +\infty$.

    If $\LA( (\ell', x), d_v-d_u-2 )$ does not exist then, using Equation~\eqref{eq:ea_query_correctness_proof} and the induction hypothesis, we have $\EA( u, v, t ) = \ell'' = \EA(x, v, \ell') = +\infty$.
\end{proof}

\end{document}